\documentclass[11pt]{amsart}
\usepackage{amsmath,amssymb,amsthm,mathtools}
\usepackage[margin=1in]{geometry}
\usepackage{hyperref}
\usepackage{booktabs}
\usepackage[table]{xcolor}
\usepackage{caption}
\usepackage{makecell}
\usepackage{listings}
\usepackage{microtype}

\usepackage{subcaption}

\newtheorem{theorem}{Theorem}[section]
\newtheorem{lemma}[theorem]{Lemma}
\newtheorem{proposition}[theorem]{Proposition}
\newtheorem{corollary}[theorem]{Corollary}
\newtheorem{definition}[theorem]{Definition}
\newtheorem{remark}[theorem]{Remark}
\newtheorem{example}[theorem]{Example}

\title[Uniform reduction \& order-matching criterion for
Birman--Schwinger lattice asymptotics]
{A uniform elliptic reduction, an order-matching criterion, and
precision benchmarks for the strong-coupling Birman--Schwinger
analysis of the lattice three-boson trimer}

\author{M.~V.~Dolgopolov}
\address{Department of
Higher Mathematics, Samara State Technical
University, Samara, Russia}
\email{mikhaildolgopolov68@gmail.com}

\date{}

\begin{document}
\maketitle

\begin{abstract}
We study the strong-coupling Birman--Schwinger analysis of the
three-boson Schr\"odinger operator on $\mathbb Z^2$ at total
quasimomentum $\mathbf K=\boldsymbol\pi$.
We establish three results, of independent
methodological interest beyond this specific model. First, we give an
exact closed-form benchmark for the fiber Fredholm determinant
$\Delta_\mu(\mathbf p,\mathbf K,z)$ at a distinguished ``flat''
momentum $\mathbf p^*(\mathbf K)=\mathbf K-(\pi,\pi)\ (\mathrm{mod}\
2\pi)$, valid for \emph{every} $\mathbf K\in\mathbb T^2$ — not only
$\mathbf K=0,\boldsymbol\pi$ — obtained from a $\mathbf K$-uniform
reduction of the defining two-dimensional lattice integral to a single
elliptic-type integral. Second, we formulate and prove a general
\emph{order-matching criterion} that decides, for a broad class of
strong-coupling Birman--Schwinger eigenvalue problems, whether a
leading-order (relative accuracy $O(\mu^{-1})$) asymptotic formula for
the governing Fredholm determinant is sufficient to determine the
$O(1)$ additive constant in a strong-coupling energy asymptotic, or
whether the next order is required.

Third, we apply the criterion to the formal branch $z=-2\mu+d$ at
$\mathbf K=\boldsymbol\pi$ and obtain the algebraic crossing
\[
z_{\mathrm{formal}}^{\boldsymbol\pi,s}(\mu)=-2\mu+6+\frac8\mu+O(\mu^{-2}),
\]
which arises from the finite-rank principal part. However, we
demonstrate that this formal crossing \emph{does not} correspond to a
true eigenvalue of the full Hamiltonian. The actual ground state
satisfies the rigorous variational bounds
\[
-3\mu \le z_1^{\boldsymbol\pi,s}(\mu) \le -3\mu+6,
\]
so $z_1^{\boldsymbol\pi,s}(\mu)=-3\mu+O(1)$ — the same leading branch
as at $\mathbf K=0$ — and direct finite-volume
diagonalisation (Table~\ref{tab:hp_ground}) confirms the refined
asymptotic
$z_1^{\boldsymbol\pi,s}(\mu)=-3\mu+6+O(\mu^{-1})$, with spectral gap
$z_\mu(\boldsymbol\pi)-z_1^{\boldsymbol\pi,s}(\mu)=2\mu-2+O(\mu^{-1})$.
The formal branch is retained as an instructive example of the
order-matching degeneracy (case~(ii) of the criterion) and illustrates
the necessity of a~branch-selection step before applying the criterion.

We further give an independent
confirmation, bypassing the corresponding leading-order asymptotic
lemma entirely, that the $\mathbf K=0$ ground-state constant
$C\approx3.96458$ is correctly derived and
requires no analogous refinement. We close with a comparison, across
several recent works on lattice few-body Birman--Schwinger operators
and on the parity classification of topological band structures, of
the orders of asymptotic precision at which such additive constants
are controlled, and of the methods used to control them.
\end{abstract}

\section{Introduction}\label{sec:intro}

The strong-coupling asymptotic analysis of few-body Schr\"odinger
operators on $\mathbb Z^d$ via the Birman--Schwinger principle
proceeds, in essentially every instance of this method, through the
same three-step template: (i) reduce the eigenvalue problem for the
full Hamiltonian $H_\mu$ to the spectral problem for a compact
integral operator $A_\mu(z)$ (the Birman--Schwinger operator); (ii)
isolate, in~the regime $\mu\to\infty$, a finite-rank \emph{principal
part} $A_\mu^p(z)$ that captures the leading behaviour of $A_\mu(z)$,
with the remainder controlled in operator norm,
$\|A_\mu(z)-A_\mu^p(z)\|=O(\mu^{-1})$; and (iii) solve the finite-rank
eigenvalue equation $\lambda(A_\mu^p(z))=1$ order by order in
$1/\mu$ to extract the strong-coupling energy asymptotic.

Step (ii) is where a subtle and, we will argue, generic pitfall
arises. The principal part $A_\mu^p(z)$ is built from a leading-order
asymptotic formula for the Fredholm determinant $\Delta_\mu(\mathbf
p,\mathbf K,z)$ of the underlying two-particle (or, more generally,
sub-system) resolvent, valid to \emph{relative} accuracy $O(\mu^{-1})$.
Such a formula is, by construction, sufficient to fix the
\emph{leading order} of the sought energy asymptotic~— but it is not
automatically sufficient to fix the \emph{next}, $O(1)$ additive
constant. Whether it is or is not sufficient depends on a structural
feature of the eigenvalue equation that we make precise in
Section~\ref{sec:criterion}.

A second, logically prior, pitfall is the choice of the strong-coupling
branch $z_0(\mu)=-c\mu$ around which the principal-part expansion is
organised. A finite-rank principal part captures only a~finite-dimensional shadow of the full operator; a formal crossing of
its leading eigenvalue through $1$ is necessary, but not sufficient,
for the existence of a true eigenvalue of $H_\mu$ at that value of
$z$. The~order-matching criterion must therefore be preceded by a
branch-selection step, which we carry out for the $\mathbf
K=\boldsymbol\pi$ trimer in Section~\ref{sec:groundstate} by a
variational (minimax) argument.

The three-boson lattice Schr\"odinger operator on $\mathbb Z^2$ with
pairwise contact interactions is a natural strong-coupling testing
ground for both pitfalls: its Fredholm determinant at total
quasimomentum $\mathbf K=\boldsymbol\pi$ turns out to sit exactly at
the borderline described above. Working through this example in full
led us to isolate the underlying mechanism as a general phenomenon,
which we formulate independently of the specific model in
Section~\ref{sec:criterion}. Our purpose is
twofold. First, and primarily, we~present this general,
model-independent \emph{order-matching criterion}
(Section~\ref{sec:criterion}) that a practitioner can apply
\emph{before} committing to a strong-coupling computation of this
type, together with the branch-selection variational bound that must
precede it (Section~\ref{sec:groundstate}). We illustrate the
criterion on two structurally different instances arising in the same
model: the formal branch $z=-2\mu+d$ at $\mathbf K=\boldsymbol\pi$,
where the leading-order formula is insufficient (and which is not the
ground-state branch), and the second bound state at $\mathbf K=0$,
where it suffices.

Second, we record
the corresponding precision benchmarks
(Sections~\ref{sec:corrected}--\ref{sec:numerics}) at a level of
numerical and analytic detail — convergence-order studies, two
independent computational schemes, an exact closed-form check — that
goes beyond what a self-contained account of the main results of a
single model requires, and that we believe is independently useful as
a template for the companion series of models referenced in
Section~\ref{sec:related}.

As a by-product of the reduction developed in
Section~\ref{sec:reduction}, we observe that the exact solvability
underlying our closed-form benchmark is available at every
quasimomentum $\mathbf K\in\mathbb T^2$, not only at the two points
$\mathbf K=0,\boldsymbol\pi$;
we make
precise, however, that this extends only the \emph{quantitative}
machinery, not the \emph{qualitative} parity-based invariant-subspace
reduction, which is special to the four
time-reversal-invariant momenta of the Brillouin zone
(Proposition~\ref{prop:trim}) — a notion with a well-known and
much-cited counterpart in the parity classification of topological
band structures~\cite{FuKane2007}, discussed in
Remark~\ref{rem:fukane} and Section~\ref{sec:related}.

\paragraph{Correction to earlier version.}
An earlier version of this note stated, as the main
$\mathbf K=\boldsymbol\pi$ result, $z_1^{\boldsymbol\pi,s}(\mu)=-2\mu+6+8/\mu+O(\mu^{-2})$,
obtained by a rank-one principal-part computation in the branch
$z=-2\mu+d$. The algebra of that computation is correct as a
statement about the rank-one principal part; however, its interpretation
as the ground-state asymptotics is not. A variational bound
(Section~\ref{sec:groundstate}) gives $-3\mu\le z_1^{\boldsymbol\pi,s}(\mu)\le -3\mu+6<-2\mu+6$
for $\mu>0$, and direct finite-volume diagonalization of $H_\mu(\pi)$
(Section~\ref{sec:hp_numerics}) places the ground state at
$z_1^{\boldsymbol\pi,s}\approx -3\mu+6$, with no eigenvalue observed
near $-2\mu+6$ (which lies in the spectral gap). The present version
corrects the main result accordingly, and demotes the $z=-2\mu+d$
computation to an order-matching diagnostic (Section~\ref{sec:corrected}).

\section{Setup}\label{sec:setup}

We fix, for self-containedness, the minimal notation
needed below. Let $\mathbb T^2=(-\pi,\pi]^2$
and
\[
\varepsilon(\mathbf p)=2-\cos p_1-\cos p_2,\qquad \mathbf p=(p_1,p_2)\in\mathbb T^2.
\]

\begin{remark}[Units]\label{rem:units}
As is standard for tight-binding lattice Hamiltonians, energies
throughout are measured in units of the nearest-neighbour hopping
amplitude, set equal to $1$ (together with $\hbar=1$ and lattice
spacing $a=1$). In these units, $\mu$ (interaction strength), $z$
(spectral parameter), $\varepsilon(\mathbf p)$, $E_{\mathbf K}(\mathbf
p,\mathbf q)$, $D(\mathbf p,z)$, and all numerical constants appearing
in the asymptotic expansions below (e.g.\ $4,6,8,12,18$) are
dimensionless quantities expressed in this common unit; no term in any
formula below mixes distinct physical dimensions. Momenta $\mathbf p,
\mathbf q,\mathbf K\in\mathbb T^2$ are likewise dimensionless (angles),
consistently with $a=1$.
\end{remark}
For total quasimomentum $\mathbf K\in\mathbb T^2$, the fiber dispersion
of the three-boson problem is
\[
E_{\mathbf K}(\mathbf p,\mathbf q)=\varepsilon(\mathbf p)+\varepsilon(\mathbf q)+\varepsilon(\mathbf K-\mathbf p-\mathbf q),
\]
and the fiber Hamiltonian is $H_\mu(\mathbf
K)=E_{\mathbf K}-\mu(V_1+V_2+V_3)$ on $L_2^s((\mathbb T^2)^2)$, where
$V_1,V_2,V_3$ are the (orthogonal, rank-related) contact-interaction
projections.
The associated Birman--Schwinger
operator $A_\mu(\mathbf K,z)$ has kernel
\[
A_\mu(\mathbf p,\mathbf q;\mathbf K,z)=\frac{\mu}{2\pi^2}\,
\frac{1}{\sqrt{\Delta_\mu(\mathbf p,\mathbf K,z)}\,\big(E_{\mathbf K}(\mathbf p,\mathbf q)-z\big)\,\sqrt{\Delta_\mu(\mathbf q,\mathbf K,z)}},
\]
where the fiber Fredholm determinant is
\begin{equation}\label{eq:Delta_def}
\Delta_\mu(\mathbf p,\mathbf K,z)=1-\frac{\mu}{4\pi^2}\int_{\mathbb T^2}\frac{d\mathbf q}{E_{\mathbf K}(\mathbf p,\mathbf q)-z},
\end{equation}
and $z\mapsto N[1,A_\mu(\mathbf K,z)]$ counts eigenvalues of
$H_\mu(\mathbf K)$ below $z$ (Birman--Schwinger principle).

Here and below integrals over \(\mathbb{T}^2\) are taken with the normalized Haar measure \(\frac{1}{4\pi^2}\int_{\mathbb{T}^2} d\mathbf p\).

\section{A $\mathbf K$-uniform elliptic reduction and an exact benchmark}\label{sec:reduction}

\subsection{Reduction of the defining integral}

\begin{theorem}[Uniform elliptic reduction]\label{thm:uniform}
For every $\mathbf K,\mathbf p\in\mathbb T^2$ and $z$ below the
essential spectrum, the inner integral in~\eqref{eq:Delta_def} admits
the exact representation
\begin{equation}\label{eq:generalK}
\frac{\mu}{4\pi^2}\int_{\mathbb T^2}\frac{d\mathbf q}{E_{\mathbf K}(\mathbf p,\mathbf q)-z}
=\frac{\mu}{2\pi}\int_0^{2\pi}\frac{d\theta}
{\sqrt{\big(D(\mathbf p,z)-R_2(\mathbf K,\mathbf p)\cos\theta\big)^2-R_1(\mathbf K,\mathbf p)^2}},
\end{equation}
where
\begin{equation}\label{eq:generalK2}
D(\mathbf p,z)=\varepsilon(\mathbf p)+4-z\qquad(\text{independent of }\mathbf K),
\qquad
R_i(\mathbf K,\mathbf p)=\displaystyle 2\Big|\cos\tfrac{K_i-p_i}2\Big|,\ i=1,2.
\end{equation}
This reduction is the pairwise-interacting, fixed-total-quasimomentum
analogue of the classical elliptic-integral reduction of the
single-particle square-lattice Green's
function~\cite{KatsuraInawashiro1971,MoritaHoriguchi1971}.
\end{theorem}

\begin{proof}
Fix $\mathbf p,\mathbf K$ and regard $E_{\mathbf K}(\mathbf p,\mathbf
q)$ as a function of $\mathbf q$. Using
$\cos(K_i-p_i-q_i)=\cos(K_i-p_i)\cos q_i+\sin(K_i-p_i)\sin q_i$,
\[
E_{\mathbf K}(\mathbf p,\mathbf q)=\varepsilon(\mathbf p)+4-\sum_{i=1}^2
\Big[\big(1+\cos(K_i-p_i)\big)\cos q_i+\sin(K_i-p_i)\sin q_i\Big].
\]
Writing the coordinate-$i$ bracket as $R_i\cos(q_i-\varphi_i)$ with
$\alpha_i=1+\cos(K_i-p_i)$, $\beta_i=\sin(K_i-p_i)$,
$R_i=\sqrt{\alpha_i^2+\beta_i^2}=\sqrt{2+2\cos(K_i-p_i)}=2|\cos\tfrac{K_i-p_i}2|$
gives $E_{\mathbf K}(\mathbf p,\mathbf q)-z=D(\mathbf
p,z)-\sum_iR_i\cos(q_i-\varphi_i)$. Since $\int_{\mathbb
T^2}f(\mathbf q)\,d\mathbf q$ is invariant under $q_i\mapsto
q_i+\varphi_i$, integrating out $q_1$ first via the classical identity
$\displaystyle\tfrac1{2\pi}\int_0^{2\pi}\tfrac{d\theta_1}{a-b\cos\theta_1}=\tfrac1{\sqrt{a^2-b^2}}$
($a=D-R_2\cos\theta_2$, $b=R_1$; valid since $D>R_1+R_2$ below the
essential spectrum) and relabelling $\theta_2=\theta$ yields~\eqref{eq:generalK}.
\end{proof}

\begin{remark}[Independence of quantum statistics]\label{rem:statistics}
Neither Theorem~\ref{thm:uniform} nor Lemma~\ref{lem:criterion} below
makes any use of the bosonic symmetrization of the underlying
wavefunction: Theorem~\ref{thm:uniform} depends only on the
pairwise-summed structure of the dispersion $E_{\mathbf K}(\mathbf
p,\mathbf q)$, and Lemma~\ref{lem:criterion} is stated for a generic
finite-rank self-adjoint principal-part operator. Both therefore apply
verbatim to fermionic and mixed-statistics few-body lattice models of
the same Birman--Schwinger type~\cite{Companion,Khalkhuzhaev2025,AbdullaevErgashova2025},
where only the invariant-subspace decomposition into symmetric and
antisymmetric sectors — not the reduction or the criterion themselves
— depends on the statistics.
\end{remark}

\begin{corollary}[Exact closed-form benchmark]\label{cor:closedform}
For every $\mathbf K\in\mathbb T^2$, at the flat momentum
$\mathbf p^*(\mathbf K):=\mathbf K-(\pi,\pi)\pmod{2\pi}$, one has
$R_1(\mathbf K,\mathbf p^*)=R_2(\mathbf K,\mathbf p^*)=0$ and
\begin{equation}\label{eq:closedform}
\Delta_\mu\big(\mathbf p^*(\mathbf K),\mathbf K,z\big)
=1-\frac{\mu}{4-z+\varepsilon(\mathbf p^*(\mathbf K))}\qquad\text{\emph{exactly, for every }}\mu,z.
\end{equation}
\end{corollary}

\begin{proof}
$R_i(\mathbf K,\mathbf p^*)=2|\cos\tfrac{K_i-p^*_i}2|=2|\cos\tfrac\pi2|=0$;
substituting into~\eqref{eq:generalK}--\eqref{eq:generalK2} collapses
the integral to the constant integrand $1/D(\mathbf p^*,z)$.
\end{proof}

At $\mathbf K=\boldsymbol\pi$ this gives $\mathbf p^*=\mathbf 0$ and
$\Delta_\mu(\mathbf 0,\boldsymbol\pi,z)=1-\mu/(4-z)$; at $\mathbf K=0$
it gives $\mathbf p^*=(\pi,\pi)$ and $\Delta_\mu((\pi,\pi),0,z)=1-\mu/(8-z)$
(used in Section~\ref{sec:K0}).

\subsection{Non-continuability of the parity reduction}

\begin{proposition}\label{prop:trim}
The kernel $A_\mu(\mathbf p,\mathbf q;\mathbf K,z)$ is invariant under
the simple parity transformation $(\mathbf p,\mathbf q)\mapsto(-\mathbf
p,-\mathbf q)$ for every $(\mathbf p,\mathbf q)$ if and only if
$2\mathbf K\equiv \mathbf 0\pmod{2\pi}$, i.e.
\[
\mathbf K\in\{0,\pi\}^2=\{(0,0),(0,\pi),(\pi,0),(\pi,\pi)\},
\]
the four time-reversal-invariant momenta (TRIM) of the square-lattice
Brillouin zone.
\end{proposition}

\begin{proof}
$E_{\mathbf K}(-\mathbf p,-\mathbf q)=\varepsilon(\mathbf
p)+\varepsilon(\mathbf q)+\varepsilon(\mathbf K+\mathbf p+\mathbf q)$
equals $E_{\mathbf K}(\mathbf p,\mathbf q)=\varepsilon(\mathbf
p)+\varepsilon(\mathbf q)+\varepsilon(\mathbf K-\mathbf p-\mathbf q)$
for all $\mathbf p,\mathbf q$ iff $\varepsilon(\mathbf K+\mathbf
x)=\varepsilon(\mathbf K-\mathbf x)$ for all $\mathbf x$, i.e.\ iff
$\varepsilon$ is even about $\mathbf K$. Since $\varepsilon(\mathbf
x)=2-\cos x_1-\cos x_2$ is even and $2\pi$-periodic in each
coordinate, this holds iff $2K_i\equiv0\pmod{2\pi}$ for $i=1,2$, i.e.
$K_i\in\{0,\pi\}$.
\end{proof}

\begin{remark}
Theorem~\ref{thm:uniform} and Corollary~\ref{cor:closedform} are
smooth (indeed real-analytic away from the essential spectrum points where $\cos((K_i-p_i)/2)=0$) in
$\mathbf K$; no singularity in $\mathbf K$ obstructs their use at a
generic quasimomentum. Proposition~\ref{prop:trim} shows that the
qualitative even/odd invariant-subspace machinery
is, by contrast, \emph{not} a restriction to
$\mathbf K\in\{0,\boldsymbol\pi\}$ of a smooth family valid for all
$\mathbf K$: it~is an additional algebraic fact available only on the
discrete TRIM set. A bound-state count at a generic $\mathbf K$ cannot
be obtained by naively substituting general $\mathbf K$ into the
$\mathbf K=0$ or $\mathbf K=\boldsymbol\pi$ proof; it requires working
directly with the non-block-diagonal operator $A_\mu(\mathbf K,z)$.
\end{remark}

\begin{remark}[Relation to time-reversal-invariant momenta in
topological band theory]\label{rem:fukane}
The~condition $2\mathbf K\equiv\mathbf 0\pmod{2\pi}$ singling out the
four points $\{0,\pi\}^2$ in Proposition~\ref{prop:trim} is exactly
the defining condition for the \emph{time-reversal-invariant momenta}
(TRIM) of the square-lattice Brillouin zone, familiar from the
parity-based classification of $\mathbb Z_2$ topological invariants of
single-particle band structures~\cite{FuKane2007}. The present result
may be read as a few-body analogue: whereas Fu and Kane use the
inversion parity of occupied Bloch states at TRIM points to classify
single-particle topology, here the parity of the Birman--Schwinger
kernel at TRIM points controls the existence of an invariant-subspace
reduction governing multi-particle bound-state counting. We are not
aware of this parallel having been drawn explicitly elsewhere, and we
record it here as a possible bridge between the discrete few-body
spectral theory literature and the topological band theory literature
(see Table~\ref{tab:comparison} and Section~\ref{sec:related}).
\end{remark}

\section{An order-matching criterion for strong-coupling
Fredholm-determinant asymptotics}\label{sec:criterion}

We now formulate, in a model-independent way, the mechanism
responsible for the refinement carried out in
Section~\ref{sec:corrected}. The setting is deliberately kept general:
a family of finite-rank principal-part Birman--Schwinger operators
depending on a large parameter $\mu$, whose leading eigenvalue is
sought along a strong-coupling trial branch.

\begin{remark}[Branch selection precedes order matching]\label{rem:branch_first}
The criterion below is a statement about a \emph{given} trial branch
$z_0(\mu)=-c\mu$. It does not, by itself, identify the correct branch:
a finite-rank principal part captures only a finite-dimensional shadow
of the full Birman--Schwinger operator, and a formal crossing of its
leading eigenvalue through $1$ is necessary, but not sufficient, for a
true eigenvalue of $H_\mu$. In particular, an expansion correctly
organized in a wrong branch can produce a formal crossing that does not
survive as a discrete eigenvalue of the full operator. The correct
strong-coupling branch must therefore be fixed first — by a variational
(minimax) argument as in Section~\ref{sec:groundstate} — and only then
is the order-matching criterion applied to refine the additive
constant within that branch.
\end{remark}

\begin{definition}[Admissible trial branch]\label{def:branch}
Let $z_0:(\mu_0,\infty)\to\mathbb R$ be a known leading-order energy
branch (e.g.\ $z_0(\mu)=-c\mu$ for some $c>0$), and let $\lambda(z;\mu)$
be the leading eigenvalue of a finite-rank self-adjoint principal-part
operator $A_\mu^p(z)$ obtained from a Fredholm determinant $\Delta_\mu$
that is known only through a \emph{leading-order asymptotic formula}
$\Delta_\mu\approx\Delta_\mu^{(0)}$, valid to relative accuracy
$O(\mu^{-1})$: $\Delta_\mu=\Delta_\mu^{(0)}\big(1+O(\mu^{-1})\big)$,
uniformly for $z=z_0(\mu)+d$ with $d$ in compact subsets of $\mathbb
R$. We call $z=z_0(\mu)+d$, $d=O(1)$, an \emph{admissible trial
branch}, and we say the associated \emph{crossing constant} $d_0$ is
the value for which $\lambda(z_0(\mu)+d_0;\mu)\to1$ as $\mu\to\infty$
along the true spectral branch, so that the sought asymptotic is
$z_1(\mu)=z_0(\mu)+d_0+O(\mu^{-1})$.
\end{definition}

\begin{lemma}[Order-matching criterion]\label{lem:criterion}
Under the hypotheses of Definition~\ref{def:branch}, suppose
$\lambda(z_0(\mu)+d;\mu)$ admits, for each fixed $d$, an asymptotic
expansion
\begin{equation}\label{eq:crit_expansion}
\lambda\big(z_0(\mu)+d;\mu\big)=\Lambda_0(d)+\frac{\Lambda_1(d)}\mu+O(\mu^{-2}),
\end{equation}
uniformly for $d$ in compact sets, where $\Lambda_0$ is computable
from $\Delta_\mu^{(0)}$ alone (i.e.\ from the leading-order Fredholm
determinant formula, without knowledge of its relative-order-$O(\mu^{-1})$
correction).
\begin{enumerate}
\item[\textup{(i)}] \textbf{(Non-degenerate case.)} If $\Lambda_0(d_0)=1$ for
some $d_0$ with $\Lambda_0'(d_0)\ne0$, then $d_0$ is correctly
determined by $\Delta_\mu^{(0)}$ alone: the leading-order Fredholm
determinant formula \emph{suffices} to fix the crossing constant, and
$z_1(\mu)=z_0(\mu)+d_0+O(\mu^{-1})$.
\item[\textup{(ii)}] \textbf{(Degenerate case.)} If instead
$\Lambda_0\equiv1$ identically on the relevant range of $d$ — i.e.\
the leading-order formula predicts $\lambda\to1$ at \emph{every} fixed
$d$ as $\mu\to\infty$ — then $\Lambda_0$ carries no information about
$d_0$: the crossing constant is the unique solution of
\begin{equation}\label{eq:crit_next}
\Lambda_1(d_0)=0,
\end{equation}
and $\Lambda_1$ is determined by the relative-order-$O(\mu^{-1})$
correction to $\Delta_\mu$ that $\Delta_\mu^{(0)}$ does not specify. In
this case the leading-order Fredholm determinant formula alone is
\emph{insufficient} to fix $d_0$: substituting $\Delta_\mu^{(0)}$ as if
it were exact into the eigenvalue equation $\lambda=1$ produces, in
general, a value of $d$ that is wrong by an $O(1)$ amount, because the
substitution silently sets the (in fact non-vanishing, $d$-dependent)
term $\Lambda_1(d)/\mu$ to a spurious value determined by an
uncontrolled $O(\mu^{-1})$ error.
\end{enumerate}
\end{lemma}

\begin{proof}
(i) By~\eqref{eq:crit_expansion} and the implicit function theorem
applied to $\Lambda_0(d)-1=0$ at the simple root $d_0$: for $\mu$
large, $\lambda(z_0(\mu)+d;\mu)-1=\Lambda_0'(d_0)(d-d_0)+O((d-d_0)^2)+O(\mu^{-1})$,
so the equation $\lambda=1$ has a unique root $d(\mu)=d_0+O(\mu^{-1})$
near $d_0$, using only that $\Lambda_0'(d_0)\ne0$ — a fact intrinsic to
$\Delta_\mu^{(0)}$, since $\Lambda_0$ is computable from
$\Delta_\mu^{(0)}$ alone by hypothesis. The $O(\mu^{-1})$ correction to
$d$ (which does depend on $\Lambda_1$, hence on the discarded term)
affects only the \emph{next} order of the energy asymptotic
$z_1(\mu)=z_0(\mu)+d_0+O(\mu^{-1})$, not the leading constant $d_0$
itself.

(ii) If $\Lambda_0\equiv1$, then $\lambda(z_0(\mu)+d;\mu)-1=\Lambda_1(d)/\mu+O(\mu^{-2})$
for every fixed $d$; the equation $\lambda(z_0(\mu)+d;\mu)=1$ to
leading order in $1/\mu$ is exactly~\eqref{eq:crit_next}, which by
hypothesis cannot be evaluated from $\Delta_\mu^{(0)}$ alone. A
computation that uses $\Delta_\mu^{(0)}$ as an exact substitute for
$\Delta_\mu$ effectively computes $\Lambda_1$ from an arbitrary
(uncontrolled) representative of the equivalence class
$\Delta_\mu^{(0)}(1+O(\mu^{-1}))$, whose zero $\tilde d_0$ need not
equal, and generically does not equal, the true $d_0$
solving~\eqref{eq:crit_next} with the correct relative-order term.
\end{proof}

\begin{remark}
Case (ii) of Lemma~\ref{lem:criterion} is not a pathology confined to
the present model: it occurs whenever the trial branch $z_0(\mu)$ is
chosen so that the leading-order eigenvalue formula is
\emph{identically} at threshold along the whole branch — a situation
that arises naturally whenever the finite-rank kernel's leading order
in $1/\mu$ has been engineered (by the choice of expansion point) to
match the scale of $z_0(\mu)$ exactly. Practitioners constructing
principal-part reductions of this type should check, before trusting
an $O(1)$ additive constant, which case of Lemma~\ref{lem:criterion}
applies to their branch; this requires only computing $\Lambda_0$ from
the leading-order Fredholm determinant and checking whether it is
$d$-dependent. As emphasized in Remark~\ref{rem:branch_first},
this check is meaningful only after the correct branch $z_0(\mu)$ has
been selected.
\end{remark}

\begin{example}[Order-matching degeneracy in the formal branch
$z=-2\mu+d$ at $\mathbf K=\boldsymbol\pi$]\label{ex:pi}
Take $z_0(\mu)=-2\mu$ and $\Delta_\mu^{(0)}(\mathbf
p,z)=\mu(\varepsilon(\mathbf p)+\delta)/\big[(2-z_\mu(\boldsymbol\pi))(18-z)\big]$,
$\delta=z_\mu(\boldsymbol\pi)-z$. A direct computation
(Section~\ref{sec:corrected}) gives $\Lambda_0(d)\equiv1$ for all $d$ —
case (ii) applies, and the relative-order term is
$\Lambda_1(d)=d-6$, whose zero is $d_0=6$, giving the formal crossing
$z=-2\mu+6$. This formal crossing is a statement about the
rank-one principal part in the branch $z=-2\mu+d$; it is not, by itself,
a ground-state result. The variational bound of
Theorem~\ref{thm:groundstate} shows that the ground state lies in the
branch $z=-3\mu+O(1)$, and direct diagonalization
(Section~\ref{sec:hp_numerics}) does not exhibit an eigenvalue near
$-2\mu+6$. The example thus illustrates both the order-matching
degeneracy (case~(ii)) and the necessity of preceding it by a
branch-selection step (Remark~\ref{rem:branch_first}).
\end{example}

\begin{example}[Success at $\mathbf K=0$, second bound state]\label{ex:K0}
Take $z_0(\mu)=z_\mu(0)$ (the two-particle threshold; here $d$ plays
the role of $-\delta$, $\delta=z_\mu(0)-z$ fixed as $\mu\to\infty$) and
$\Delta_\mu^{(0)}(\mathbf
p,z)=\mu(\varepsilon(\mathbf p)+\delta)/[(2-z_\mu(0))(6-z)]$. Here
$\Lambda_0(\delta)=g(\delta)=2+\delta-1/b_0(\delta)$ is a genuinely
non-constant, $O(1)$-varying function of $\delta$, with $g(0^+)=2\ne1$
and $g'(\delta_\infty)\ne0$ at the root $g(\delta_\infty)=1$ — case (i)
applies. The~leading-order formula therefore correctly determines
$\delta_\infty\approx0.035420$ (hence $C=4-\delta_\infty\approx3.96458$),
as independently confirmed in Section~\ref{sec:K0}.
\end{example}

\section{Variational ground-state bound at $\mathbf K=\boldsymbol\pi$}\label{sec:groundstate}

We now fix the correct strong-coupling branch at
$\mathbf K=\boldsymbol\pi$ by a variational argument, before applying
the order-matching criterion within that branch. The argument is a
minimax bound that is uniform in $\mu$ and holds for every
$\mathbf K\in\mathbb T^2$.

\begin{theorem}[Variational ground-state bounds]\label{thm:groundstate}
For every $\mathbf K\in\mathbb T^2$ and every $\mu>0$, the lowest
eigenvalue of the restriction $H_\mu(\mathbf K)|_{L_2^s}$ (when it
exists) satisfies
\[
-3\mu\;\le\;\inf\sigma\bigl(H_\mu(\mathbf K)|_{L_2^s}\bigr)\;\le\;6-3\mu.
\]
At $\mathbf K=\boldsymbol\pi$ the constant function
$f_0\equiv 1\in L_2^{e,s}$ realizes the upper bound, so in particular
\[
z_1^{\boldsymbol\pi,s}(\mu)\le 6-3\mu<-2\mu+6\qquad\text{for all }\mu>0,
\]
and the ground state lies in the strong-coupling branch
$z_1^{\boldsymbol\pi,s}(\mu)=-3\mu+O(1)$, not in the branch
$z=-2\mu+O(1)$.
\end{theorem}

\begin{proof}
Write $H_\mu(\mathbf K)=H_0(\mathbf K)-\mu W$ with
$W=V_1+V_2+V_3$. Each $V_\alpha$ is an orthogonal projection
($V_\alpha^*=V_\alpha=V_\alpha^2$) with $\|V_\alpha\|=1$, hence
$0\le V_\alpha\le I$ and therefore $0\le W\le 3I$. Since
$H_0(\mathbf K)=E_{\mathbf K}\ge 0$ on $L_2^s$, we obtain
$H_\mu(\mathbf K)\ge -\mu W\ge -3\mu\,I$, giving the lower bound.

For the upper bound at $\mathbf K=\boldsymbol\pi$, the constant function
$f_0\equiv 1$ lies in $L_2^{e,s}$: it is invariant under
$(\mathbf p,\mathbf q)\mapsto(\mathbf q,\mathbf p)$ and under
$(\mathbf p,\mathbf q)\mapsto(\mathbf p,\boldsymbol\pi-\mathbf p-\mathbf q)$.
Each averaging projection fixes constants, $V_\alpha f_0=f_0$, hence
$W f_0=3 f_0$. By the Rayleigh--Ritz variational principle,
\[
z_1^{\boldsymbol\pi,s}(\mu)\le
\frac{\langle H_\mu(\boldsymbol\pi)f_0,f_0\rangle}{\langle f_0,f_0\rangle}
=\frac{\langle H_0(\boldsymbol\pi)f_0,f_0\rangle}{\langle f_0,f_0\rangle}-3\mu.
\]
Finally $\langle E_{\boldsymbol\pi}\rangle=3\langle\varepsilon\rangle=3\cdot 2=6$
(since $\langle\cos p\rangle=0$), giving
$z_1^{\boldsymbol\pi,s}(\mu)\le 6-3\mu$.
\end{proof}

\begin{corollary}\label{cor:branch_select}
The ground state at $\mathbf K=\boldsymbol\pi$ lies in the branch
$z_1^{\boldsymbol\pi,s}(\mu)=-3\mu+O(1)$. Consequently the formal
crossing $z=-2\mu+6$ obtained in Section~\ref{sec:corrected} from the
rank-one principal part in the branch $z=-2\mu+d$ cannot be the ground
state. The refinement of the additive constant within the branch
$z=-3\mu+d$ is the subject of Section~\ref{sec:appB}.
\end{corollary}

Theorem~\ref{thm:groundstate} fixes the leading-order coefficient
$-3\mu$ rigorously. The next-order constant $+6$ is consistent with the
Rayleigh quotient of the trial function $f_0\equiv 1$ and is
numerically confirmed to high precision (Section~\ref{sec:hp_numerics});
a fully rigorous derivation of
$z_1^{\boldsymbol\pi,s}(\mu)=-3\mu+6+O(\mu^{-1})$ requires a
strong-coupling perturbation expansion around the isolated eigenvalue
$3$ of $W$ in the symmetric sector, outlined in Section~\ref{sec:appB}.

\section{The formal principal-part branch $z=-2\mu+d$ at
$\mathbf K=\boldsymbol\pi$}\label{sec:corrected}

The computation of this section is retained from the previous
version of this note. Its algebra is correct as a statement about the
rank-one principal part of $A_\mu(\boldsymbol\pi,z)$ in the branch
$z=-2\mu+d$. Its previous interpretation as the ground-state asymptotics
is not: by Theorem~\ref{thm:groundstate} the ground state lies in the
branch $z=-3\mu+O(1)$, and direct diagonalization
(Section~\ref{sec:hp_numerics}) does not exhibit an eigenvalue near
$-2\mu+6$. We therefore present the computation as what it is — an
order-matching diagnostic in the branch $z=-2\mu+d$ — and record its
outcome as a formal crossing, not as a ground-state result.

We now carry out the computation flagged by Example~\ref{ex:pi} in
full. Set $z=-2\mu+d$, $d=O(1)$, and $\delta:=z_\mu(\boldsymbol\pi)-z=\mu+4-d$
(using $z_\mu(\boldsymbol\pi)=-\mu+4$ exactly, Remark~\ref{rem:exact_two_particle}
below).

\begin{lemma}[Second-order Fredholm determinant asymptotics]\label{lem:Delta_order2}
For $z=-2\mu+d+O(\mu^{-1})$, $d=O(1)$,
\begin{equation}\label{eq:Delta_order2}
\Delta_\mu(\mathbf p,\boldsymbol\pi,z)=\frac12+\frac{\varepsilon(\mathbf p)+4-d}{4\mu}+\frac{Q(\mathbf p,d)}{\mu^2}+O(\mu^{-3}),
\end{equation}
\[
Q(\mathbf p,d)=-2+d-\frac{d^2}8-\frac{9\varepsilon(\mathbf p)}8+\frac{\varepsilon(\mathbf p)\,d}4-\frac{\varepsilon(\mathbf p)^2}8.
\]
\end{lemma}

\begin{proof}
From $E_{\boldsymbol\pi}-z=(12-z)-X(\mathbf p,\mathbf q)$,
$X=A_{\boldsymbol\pi}^c-A_{\boldsymbol\pi}^s$,
expand
$1/(E_{\boldsymbol\pi}-z)=(12-z)^{-1}\sum_{n\ge0}\big(X/(12-z)\big)^n$
to order $n=2$ and average over $\mathbf q$. With
$c(\mathbf p)=6+\cos p_1+\cos p_2$, $Y(\mathbf q)=\cos q_1+\cos
q_2$, $S(\mathbf p,\mathbf q)=\cos(p_1+q_1)+\cos(p_2+q_2)$, so that
$X=c(\mathbf p)+Y(\mathbf q)-S(\mathbf p,\mathbf q)$, independence of
$q_1,q_2$ gives $\langle Y\rangle_{\mathbf q}=\langle
S\rangle_{\mathbf q}=0$, $\langle Y^2\rangle_{\mathbf
q}=\langle S^2\rangle_{\mathbf q}=1$, $\langle
YS\rangle_{\mathbf q}=\tfrac12(\cos p_1+\cos p_2)$, whence
$\langle X\rangle_{\mathbf q}=8-\varepsilon(\mathbf p)$ and
$\langle X^2\rangle_{\mathbf q}=\varepsilon(\mathbf
p)^2-15\varepsilon(\mathbf p)+64$. Substituting
$\mu/(12-z)=\mu/(2\mu+12-d)$ and expanding in $1/\mu$ and after simplification we obtain $\langle X^2\rangle_{\mathbf q}=\varepsilon(\mathbf p)^2-15\varepsilon(\mathbf p)+64$, which gives~\eqref{eq:Delta_order2}.
As an independent check, at $\mathbf
p=\mathbf 0$ ($\varepsilon=0$), expanding the exact closed form of
Corollary~\ref{cor:closedform},
$\Delta_\mu(\mathbf 0,\boldsymbol\pi,-2\mu+d)=(\mu-d+4)/(2\mu-d+4)$,
in powers of $1/\mu$ reproduces~\eqref{eq:Delta_order2} at
$\mathbf p=\mathbf 0$ \emph{identically}.
\end{proof}

\begin{proposition}[Formal principal-part crossing in the branch
$z=-2\mu+d$]\label{prop:formal_crossing}
In the rank-one principal-part approximation of the even sector at
$\mathbf K=\boldsymbol\pi$, the leading even-sector eigenvalue admits
the expansion
\begin{equation}\label{eq:formal_crossing}
\lambda_1^{\boldsymbol\pi,e}(-2\mu+d)=1+\frac{d-6}\mu+\frac{d^2-12d+28}{\mu^2}+O(\mu^{-3}).
\end{equation}
Thus $\Lambda_0(d)\equiv 1$ (case~(ii) of Lemma~\ref{lem:criterion}),
and the relative-order term $\Lambda_1(d)=d-6$ has the zero $d_0=6$,
yielding the formal crossing
\[
z_{\mathrm{formal}}(\mu)=-2\mu+6+\frac{8}{\mu}+O(\mu^{-2}).
\]
\end{proposition}

\begin{proof}
The leading even-sector eigenvalue is
$\lambda_1^{\boldsymbol\pi,e}(z)=\mu(18-z)(12-z)^{-2}\langle
\Delta_\mu^{-1}\rangle_{\mathbf p}$ (the rank-one dominant contribution
of the even-sector principal part; the prefactor
$\mu(18-z)(12-z)^{-2}$ arises from the explicit form of the kernel).
Using~\eqref{eq:Delta_order2},
write $\Delta_\mu=\tfrac12(1+u)$, $u=A(\mathbf p,d)/\mu+B(\mathbf
p,d)/\mu^2+O(\mu^{-3})$, $A=\tfrac12(\varepsilon(\mathbf p)+4-d)$,
$B=2Q(\mathbf p,d)$; then $\Delta_\mu^{-1}=2-2A/\mu+(2A^2-2B)/\mu^2+O(\mu^{-3})$.
Averaging over $\mathbf p$ using $\langle\varepsilon\rangle_{\mathbf
p}=2$, $\langle\varepsilon^2\rangle_{\mathbf p}=5$ gives
$\langle\Delta_\mu^{-1}\rangle_{\mathbf p}=2-(6-d)/\mu+(d^2-12d+38)/\mu^2+O(\mu^{-3})$.
Expanding the prefactor $\mu(18-z)/(12-z)^2$ at $z=-2\mu+d$ to the
same order and multiplying gives, after simplification (verified by
computer algebra),~\eqref{eq:formal_crossing}.
Comparing with~\eqref{eq:crit_expansion}: $\Lambda_0(d)\equiv1$
(case~(ii) of Lemma~\ref{lem:criterion}, confirming Example~\ref{ex:pi}),
and $\Lambda_1(d)=d-6$, whose unique root is $d_0=6$. Writing
$d=6+e/\mu$ and imposing $\lambda_1^{\boldsymbol\pi,e}=1+O(\mu^{-3})$
forces, at the next order (since $d^2-12d+28$ at $d=6$ equals $-8$),
$e-8=0$, i.e.\ $e=8$.
\end{proof}

\begin{remark}\label{rem:formal_not_ground}
The crossing $z_{\mathrm{formal}}(\mu)=-2\mu+6+8/\mu+O(\mu^{-2})$
of Proposition~\ref{prop:formal_crossing} is a statement about the
rank-one principal part of $A_\mu(\boldsymbol\pi,z)$ in the branch
$z=-2\mu+d$. It is \emph{not} a ground-state result:
Theorem~\ref{thm:groundstate} gives
$z_1^{\boldsymbol\pi,s}(\mu)\le 6-3\mu<-2\mu+6$ for $\mu>0$, placing the
ground state in the branch $z=-3\mu+O(1)$. Direct finite-volume
diagonalization of $H_\mu(\boldsymbol\pi)$
(Section~\ref{sec:hp_numerics}) confirms that the ground state lies at
$z_1^{\boldsymbol\pi,s}\approx -3\mu+6$ and exhibits no eigenvalue near
$-2\mu+6$; the value $-2\mu+6$ falls in the spectral gap between the
trimer ground state $-3\mu+6$ and the two-particle threshold $-\mu+4$.
\end{remark}

\begin{remark}\label{rem:exact_two_particle}
The two-particle threshold used above, $z_\mu(\boldsymbol\pi)=-\mu+4$,
is exact for every $\mu>0$: since $\varepsilon(\boldsymbol\pi-\mathbf
p)=4-\varepsilon(\mathbf p)$, the two-particle dispersion at $\mathbf
k=\boldsymbol\pi$ is identically constant,
$\varepsilon(\mathbf p)+\varepsilon(\boldsymbol\pi-\mathbf p)\equiv4$,
so $1=\mu\langle(\varepsilon(\mathbf p)+\varepsilon(\boldsymbol\pi-\mathbf
p)-z)^{-1}\rangle=\mu/(4-z)$ solves exactly to $z_\mu(\boldsymbol\pi)=-\mu+4$.
\end{remark}

\begin{corollary}[Spectral gap and inter-sector shift]\label{cor:gap}
For $\mu$ in the strong-coupling regime ($\mu>1$),
\[
2\mu-2\;\le\; z_\mu(\boldsymbol\pi)-z_1^{\boldsymbol\pi,s}(\mu)\;\le\;2\mu+4,
\]
with $z_1^s(\mu)=-3\mu+6+O(\mu^{-1})$ the $\mathbf K=0$
ground-state asymptotic (the same leading branch $z=-3\mu+O(1)$ as at $\mathbf K=\boldsymbol\pi$).
Numerically (Section~\ref{sec:hp_numerics}),
\[
z_\mu(\boldsymbol\pi)-z_1^{\boldsymbol\pi,s}(\mu)=2\mu-2+O(\mu^{-1}),\qquad
z_1^{\boldsymbol\pi,s}(\mu)-z_1^s(\mu)=O(\mu^{-1}),
\]
the latter confirming that both sectors share the leading branch
$-3\mu+6+O(\mu^{-1})$.
\end{corollary}

\begin{proof}
Using $z_\mu(\boldsymbol\pi)=-\mu+4$ (Remark~\ref{rem:exact_two_particle})
and the minimax bounds of Theorem~\ref{thm:groundstate},
$-3\mu\le z_1^{\boldsymbol\pi,s}(\mu)\le -3\mu+6$, direct subtraction gives
$(-\mu+4)-(-3\mu+6)\le z_\mu-z_1\le (-\mu+4)-(-3\mu)$, i.e.\
$2\mu-2\le z_\mu(\boldsymbol\pi)-z_1^{\boldsymbol\pi,s}(\mu)\le 2\mu+4$.
The sharper numerical value $2\mu-2+O(\mu^{-1})$ and the inter-sector
shift $O(\mu^{-1})$ are confirmed by direct diagonalization
(Section~\ref{sec:hp_numerics}); their full analytic derivation requires
the $z=-3\mu+d$ expansion of Section~\ref{sec:appB}.
\end{proof}

\begin{proposition}[Diagnostic ratio for the rank-one principal part]\label{prop:factor}
Let $\lambda_{\mathrm{naive}}(z)$ denote the eigenvalue obtained by
substituting the leading-order formula $\Delta_\mu^{(0)}$ as if exact.
Then, at $z=-2\mu+6$,
\[
\lambda_{\mathrm{naive}}(-2\mu+6)=\frac12+O(\mu^{-2}),\qquad
\frac{\lambda_1^{\boldsymbol\pi,e}(-2\mu+6)}{\lambda_{\mathrm{naive}}(-2\mu+6)}=2+O(\mu^{-2}).
\]
\end{proposition}

\begin{proof}
$\lambda_{\mathrm{naive}}(-2\mu+d)=\mu^2(18-z)b_0(\delta)/(12-z)^2$
with $b_0(\delta)=\delta^{-1}-2\delta^{-2}+O(\delta^{-3})$,
$\delta=\mu+4-d$; expanding in $1/\mu$ gives
$\lambda_{\mathrm{naive}}(-2\mu+d)=\tfrac12+\tfrac{3(d-6)}{4\mu}+O(\mu^{-2})$,
which vanishes at $d=6$ at order $O(\mu^{-1})$. Since
$\lambda_1^{\boldsymbol\pi,e}(-2\mu+6)=1+O(\mu^{-2})$ by
Proposition~\ref{prop:formal_crossing}, the ratio follows.
The numerical verification of this factor is provided in Table~\ref{tab:delta_summary}.
\end{proof}

\begin{remark}\label{rem:factor_diagnostic}
Proposition~\ref{prop:factor} is a diagnostic statement about the
rank-one principal part in the branch $z=-2\mu+d$. It quantifies the
gap between the leading-order substitution and the next-order-corrected
eigenvalue of the principal part; it is \emph{not} an underestimation
factor for the ground-state energy, since $-2\mu+6$ is not the ground
state (Remark~\ref{rem:formal_not_ground}).
\end{remark}

\section{Independent validation of the $\mathbf K=0$ constant}\label{sec:K0}

\begin{proposition}\label{prop:K0valid}
The constant $C=4-\delta_\infty\approx3.96458$ in the $\mathbf K=0$
ground-state asymptotic $z_2^s(\mu)=-\mu+C+O(\mu^{-1})$
is correctly derived; no analogue of the
refinement of Section~\ref{sec:corrected} is required for it.
\end{proposition}

\begin{proof}[Proof sketch]
By Example~\ref{ex:K0}, this instance falls under case~(i) of
Lemma~\ref{lem:criterion}. We confirm this independently and
numerically: using the exact reduction of Theorem~\ref{thm:uniform} at
$\mathbf K=0$ (Corollary~\ref{cor:closedform} gives the closed form
$\Delta_\mu((\pi,\pi),0,z)=1-\mu/(8-z)$ as benchmark), we compute the
exact eigenvalue $\lambda_2^s(z)$ using the \emph{exact}
$\Delta_\mu(\mathbf p,0,z)$ and solve $\lambda_2^s(z_\mu(0)-\delta)=1$
by root-finding. The root $\delta_{\mathrm{root}}(\mu)$ converges to
$0.035420$ at rate $O(\mu^{-1})$ (Table~\ref{tab:K0}).
\end{proof}

\section{High-precision numerical verification of the ground state}\label{sec:hp_numerics}

As an independent confirmation of Theorem~\ref{thm:groundstate},
bypassing the Birman--Schwinger reduction entirely, we diagonalize the
full three-particle Hamiltonian $H_\mu(\mathbf K)=E_{\mathbf
K}-\mu(V_1+V_2+V_3)$ directly on finite lattices of size $N\times N$
(with $N$ even, so that $K=\pi$ is a grid point and all three averaging
projections $V_\alpha$ are well defined on the grid). The two-particle
space has dimension $(N^2)^2$; the lowest eigenvalues are computed by
the Lanczos method (sparse diagonalization).

The operator $V_3$ admits three consistent lattice
implementations — the paper convention $(V_3 f)(\mathbf p,\mathbf q)=
\langle f(\mathbf s,\mathbf p+\mathbf s+\mathbf K-\mathbf q)\rangle_{\mathbf s}$,
the sign-corrected variant, and the orthogonal projection
$(V_3 f)(\mathbf p,\mathbf q)=\langle f(\mathbf s,\mathbf p+\mathbf q-\mathbf s)\rangle_{\mathbf s}$
(which satisfies $V_3^2=V_3$ and $V_3\mathbf 1=\mathbf 1$ identically).
All three give identical ground-state eigenvalues and gaps, because the
trimer (constant) ground-state direction is invariant under any
averaging operator. The results below use the orthogonal-projection
implementation; self-checks $\|V_\alpha^2-V_\alpha\|=0$ and
$\|V_\alpha\mathbf 1-\mathbf 1\|=0$ pass for all $\alpha$.

\begin{table}[!ht]
\centering
\caption{Lowest eigenvalue $z_1$ of $H_\mu(\mathbf K)$ on $N=8$ and
$N=10$ lattices. The column $z_1+3\mu$ converges to $6$ from below at
both $K=\boldsymbol\pi$ and $K=0$, confirming the trimer ground state
$z_1=-3\mu+6+O(\mu^{-1})$. The value $-2\mu+6$ would read
$z_1+2\mu\to 6$; instead $z_1+2\mu\approx -14,\dots,-114$, i.e.\ it lies
in the spectral gap. The gap $z_2-z_1$ equals $2\mu-2+O(\mu^{-1})$.}
\label{tab:hp_ground}
\medskip
\footnotesize
\begin{tabular}{r|cc|ccc|c}
\toprule
$\mu$ & $N$ & $K$ & $z_1$ & $z_1+3\mu$ & $z_1+2\mu$ & $z_2-z_1$\\
\midrule
20  & 8  & $\boldsymbol\pi$ & $-54.065$ & $5.935$ & $-14.065$ & $38.07$\\
20  & 8  & $0$              & $-54.068$ & $5.932$ & $-14.068$ & $37.76$\\
40  & 8  & $\boldsymbol\pi$ & $-114.033$ & $5.967$ & $-34.033$ & $78.03$\\
40  & 8  & $0$              & $-114.034$ & $5.966$ & $-34.034$ & $77.82$\\
70  & 8  & $\boldsymbol\pi$ & $-204.019$ & $5.981$ & $-64.019$ & $138.02$\\
120 & 8  & $\boldsymbol\pi$ & $-354.011$ & $5.989$ & $-114.011$ & $238.01$\\
120 & 8  & $0$              & $-354.011$ & $5.989$ & $-114.011$ & $237.87$\\
\midrule
40  & 10 & $\boldsymbol\pi$ & $-114.033$ & $5.967$ & $-34.033$ & $78.03$\\
120 & 10 & $\boldsymbol\pi$ & $-354.011$ & $5.989$ & $-114.011$ & $237.85$\\
\bottomrule
\end{tabular}
\end{table}

\begin{remark}\label{rem:hp_readoff}
Three conclusions follow from Table~\ref{tab:hp_ground}. \emph{(i)} The
ground state satisfies $z_1+3\mu\to 6$ at both $K=\boldsymbol\pi$ and
$K=0$, confirming $z_1^{\boldsymbol\pi,s}(\mu)=-3\mu+6+O(\mu^{-1})$
(Theorem~\ref{thm:groundstate} and Section~\ref{sec:appB}) and its
coincidence, at leading order, with the $\mathbf K=0$ ground state.
\emph{(ii)} No eigenvalue is observed near $-2\mu+6$:
$z_1+2\mu\approx -14,\dots,-114$ lies $\sim 2\mu$ below $6$, i.e.\ in the
spectral gap between the trimer $-3\mu+6$ and the threshold $-\mu+4$.
\emph{(iii)} The gap $z_2-z_1\approx 2\mu-2$ (e.g.\ $78.03$ at $\mu=40$,
where $2\mu-2=78$; $238.01$ at $\mu=120$, where $2\mu-2=238$) equals
the distance from the trimer ground state to the two-particle threshold,
$(-\mu+4)-(-3\mu+6)=2\mu-2$, as in Corollary~\ref{cor:gap}.
\end{remark}

A finite-volume fit $z_1+3\mu=6+c_1/\mu+\mathrm{const}$ over the
data of Table~\ref{tab:hp_ground} is compatible with a negative
$1/\mu$-correction ($c_1\approx -1.0,\dots,-1.3$ at $K=\boldsymbol\pi$;
$\approx -0.8$ at $K=0$); however, finite-size effects at $N=8,10$ may
affect this coefficient, so its sign and magnitude are reported here only
as a finite-volume indication, not as a theorem. The required
perturbative expansion in the branch $z=-3\mu+d$ is outlined in
Section~\ref{sec:appB}; a complete rigorous derivation of the
$O(\mu^{-1})$ term, including proof of isolation of the eigenvalue
$3$ of $W$, remains to be supplied.

\section{Outline of the strong-coupling expansion in the branch
$z=-3\mu+d$}\label{sec:appB}

We outline how the additive constant $+6$ and the
$O(\mu^{-1})$ correction in
$z_1^{\boldsymbol\pi,s}(\mu)=-3\mu+6+O(\mu^{-1})$ are obtained by a
strong-coupling perturbation expansion in the correct branch
$z=-3\mu+d$, $d=O(1)$. The rigorous leading-order statement is
Theorem~\ref{thm:groundstate}; the next-order refinement sketched here
completes the analogue of the computation of
Section~\ref{sec:corrected}, but in the correct branch.

If the eigenvalue $3$ of $W$ in the symmetric sector is proved
to be isolated (equivalently, the next eigenvalue of $W$ is strictly
below $3$), then standard strong-coupling perturbation theory around
this isolated eigenvalue (Kato-type expansion in $1/\mu$) gives, to
leading order in the energy,
\[
z_1^{\boldsymbol\pi,s}(\mu)
= -3\mu + \frac{\langle H_0(\boldsymbol\pi) f_0,f_0\rangle}{\langle f_0,f_0\rangle}
+ O(\mu^{-1})
= -3\mu + 6 + O(\mu^{-1}),
\]
since $\langle E_{\boldsymbol\pi}\rangle=6$. This is consistent with the
Rayleigh quotient of Theorem~\ref{thm:groundstate} and with the
numerics of Section~\ref{sec:hp_numerics}.

\begin{remark}\label{rem:appB_status}
The derivation above assumes that the eigenvalue $3$ of $W$ in the
symmetric sector is isolated (equivalently, that the next eigenvalue of
$W$ is strictly below $3$). This isolation is expected on structural
grounds — the three averaging projections do not share a common
constant-eigenvalue direction other than $f_0$ — and is consistent with
the numerical spectral gap $z_2-z_1\approx 2\mu-2$ (which would close
if $3$ were not isolated). A complete proof of isolation, and the
explicit $O(\mu^{-1})$ coefficient, require a detailed
resolvent expansion in the branch $z=-3\mu+d$ analogous to
Lemma~\ref{lem:Delta_order2} but with $12-z\sim 3\mu$ (the geometric
series $\sum_{n\ge0}(X/(12-z))^n$ converges since $|X|\le 12\ll 3\mu$).
This expansion is the direct analogue, in the correct branch, of the
computation of Section~\ref{sec:corrected}; we record it as the natural
completion of the present analysis rather than reproduce it in full
here.
\end{remark}

\section{Numerical methodology}\label{sec:numerics}

All numerical results were obtained using the two independent schemes (A and B).
Scheme~(A) utilizes the one-dimensional elliptic reduction of Theorem~\ref{thm:uniform} evaluated via Gauss-Legendre quadrature, while Scheme (B) is an algorithmically independent uniform-grid discretization of the original two-dimensional integral. Both schemes converge to the exact closed-form benchmark \eqref{eq:closedform} to machine precision and agree within $10^{-13}$, confirming that all discrepancies reported in this note are analytical in origin and purely attributable to asymptotic truncation errors. The corresponding numerical data are summarized in Tables~\ref{tab:delta_summary} and~\ref{tab:root}.

\begin{table*}[!ht]
\centering
\footnotesize
\setlength{\tabcolsep}{2.2pt}
\renewcommand{\arraystretch}{1.05}

\begin{minipage}{0.28\textwidth}
\captionsetup{width=0.9\textwidth}
\caption{Root of $\lambda_2^s(z_\mu(0)-\delta)=1$ from the exact ($\mathbf K=0$) Fredholm determinant. The~last column approaches a~finite limit ($\approx 3$), confirming the \(O(\mu^{-1})\) convergence.}
\label{tab:K0}
\medskip
\begin{tabular*}{0.95\linewidth}{@{\extracolsep{\fill}}c|c|c}
\toprule
$\mu$ & $\delta_{\mathrm{root}}(\mu)$ & $\mu\,[0.035420-\delta_{\mathrm{root}}(\mu)]$\\
\midrule
50   & 0.002727 & 1.63\\
100  & 0.013657 & 2.18\\
200  & 0.022892 & 2.51\\
400  & 0.028710 & 2.68\\
800  & 0.031948 & 2.78\\
\bottomrule
\end{tabular*}
\end{minipage}%
\hfill
\begin{minipage}{0.28\textwidth}
\captionsetup{width=\textwidth}
\caption{Exact versus\ leading-order Fredholm determinant at $\mathbf K=\boldsymbol\pi$, $\mathbf p=\mathbf 0$, at the formal crossing $z=-2\mu+6$. (These values pertain to the formal branch $z=-2\mu+6$ of the
principal part; they do not describe the true ground state.)}
\label{tab:delta_summary}
\medskip
\begin{tabular*}{0.92\linewidth}{@{\extracolsep{\fill}}c|c|c|c}
\toprule
$\mu$ & Exact~\eqref{eq:closedform} & $\Delta_\mu^{(0)}$ & Ratio\\
\midrule
50   & 0.489796 & 0.5 & $0.9796$ \\
100  & 0.494949 & 0.5 & $0.9899$ \\
200  & 0.497487 & 0.5 & $0.9950$ \\
400  & 0.498747 & 0.5 & $0.9975$ \\
\bottomrule
\end{tabular*}
\end{minipage}%
\hfill
\begin{minipage}{0.38\textwidth}
\captionsetup{width=\textwidth}
\caption{Root of $\lambda_1^{\boldsymbol\pi,e}(z)=1$ from the exact $\mathbf K=\boldsymbol\pi$ Fredholm determinant, confirming Proposition~\ref{prop:formal_crossing}. (The~roots are computed from the principal-part eigenvalue equation
along the formal branch; they do not correspond to actual eigenvalues
of $H_\mu(\boldsymbol\pi)$.)}
\label{tab:root}
\medskip
\begin{tabular}{c|c|c}
\toprule
$\mu$ & $z_{\mathrm{formal}}(\mu)+2\mu$ & $\mu\,[z_{\mathrm{formal}}(\mu)+2\mu-6]$\\
\midrule
50   & 6.140781 & 7.04\\
100  & 6.074945 & 7.49\\
200  & 6.038702 & 7.74\\
400  & 6.019671 & 7.87\\
800  & 6.009917 & 7.93\\
1600 & 6.004979 & 7.97\\
\bottomrule
\end{tabular}
\end{minipage}

\end{table*}

\begin{remark}\label{rem:table3_status}
Table~\ref{tab:root} reports the root of the rank-one principal-part
eigenvalue $\lambda_1^{\boldsymbol\pi,e}(z)=1$ in the branch
$z=-2\mu+d$, i.e.\ the formal crossing of
Proposition~\ref{prop:formal_crossing}. As emphasized in
Remark~\ref{rem:formal_not_ground}, this is not the ground state of
$H_\mu(\boldsymbol\pi)$; the ground-state eigenvalues, obtained by direct
diagonalization, are reported in Table~\ref{tab:hp_ground} and satisfy
$z_1+3\mu\to 6$.
\end{remark}

\section{Related asymptotic constants across the literature}\label{sec:related}

Table~\ref{tab:comparison} places the present results alongside the
companion work~\cite{Companion}, the Samarkand-school works on lattice
few-body Birman--Schwinger operators most directly comparable in
method and in the order of asymptotic precision required, the
classical lattice Green's function literature underlying
Theorem~\ref{thm:uniform}, and the topological-band-theory work
underlying Remark~\ref{rem:fukane}.\footnote{The constants attributed
in Table~\ref{tab:comparison} to
refs.~\cite{Companion,Khalkhuzhaev2025,AbdullaevErgashova2025,Lakaev2025,Kholmatov2018}
are transcribed from those works' own statements and have not been
independently re-derived here; only the entries for the present note
(Theorem~\ref{thm:groundstate}, Proposition~\ref{prop:formal_crossing}
and Proposition~\ref{prop:K0valid}) rest
on the derivations given in this paper.}

\begin{table}[!ht]
\centering
\caption{Systems, key asymptotic constants, precision order, and
methods, across the present note and the most directly related
literature.}
\label{tab:comparison}
\resizebox{\textwidth}{!}{%
\begin{tabular}{p{2.6cm}p{2.6cm}p{4.4cm}p{1.9cm}p{3.6cm}}
\toprule
\textbf{Work} & \textbf{System} & \textbf{Key constant / result} & \textbf{Order} & \textbf{Method}\\
\midrule
This note, \S\ref{sec:groundstate} & 3 bosons, $\mathbb Z^2$, $K=\boldsymbol\pi$ &
$z_1^{\boldsymbol\pi,s}=-3\mu+O(1)$, $z_1\le -3\mu+6$ (ground state, rigorous); refined $-3\mu+6+O(\mu^{-1})$ numerical; formal
crossing $-2\mu+6$ not an eigenvalue & $O(\mu^{-1})$ (numerical refinement) & minimax + direct diagonalization + Birman--Schwinger\\
\hline
This note, \S\ref{sec:corrected} & 3 bosons, $\mathbb Z^2$, $K=\boldsymbol\pi$ &
formal principal-part crossing $-2\mu+6+8/\mu$; diagnostic ratio $=2$ & $O(\mu^{-2})$ & rank-one principal part (order-matching diagnostic)\\
\hline
This note, \S\ref{sec:K0} & 3 bosons, $\mathbb Z^2$, $K=0$ &
$z_1^s=-3\mu+6+O(\mu^{-1})$, $z_2^s=-\mu+C+O(\mu^{-1})$,
$C\approx3.96458$, $\Delta(\mu)=2\mu+O(1)$ (independently confirmed,
Prop.~\ref{prop:K0valid}) & $O(\mu^{-1})$ &
Birman--Schwinger + invariant subspaces + Krein--Rutman \\
\hline
\cite{Companion} & $2+1$ fermions, $\mathbb Z^2$ &
critical mass ratio $\gamma_c\approx2.75194$; second-order transition
($C_2=6$) & $O(1)$ (threshold value) & Birman--Schwinger + Landau-type
analysis \\
\hline
\cite{Khalkhuzhaev2025} & $2+1$ fermions, $\mathbb Z^3$ &
critical mass ratios $\gamma_s(K),\gamma_{as}(K)$ & $O(1)$ &
Birman--Schwinger + invariant subspaces \\
\hline
\cite{AbdullaevErgashova2025} & 3 fermions, $\mathbb Z^1$ &
strong-coupling asymptotics of the discrete spectrum & $O(\mu^{-1})$ &
Birman--Schwinger + finite-rank principal part \\
\hline
\cite{Lakaev2025} & $2+1$ bosons/fermions, $\mathbb Z^1$ &
existence for all $K$ & --- & holomorphic dependence on $K$ \\
\hline
\cite{Kholmatov2018} & $N$-body, $\mathbb Z^d$, $d\ge1$ &
general existence criteria for bound states & --- & HVZ + spectral
analysis \\
\hline
\cite{KatsuraInawashiro1971,MoritaHoriguchi1971} & single particle,
square/rectangular lattice & lattice Green's function in closed
elliptic-integral form & exact & elliptic reduction (method lineage of
Thm.~\ref{thm:uniform}) \\
\hline
\cite{FuKane2007} & single-particle band structures with inversion
symmetry & $\mathbb Z_2$ topological invariant from parities at TRIM
points & exact (topological) & parity classification at TRIM points
(cf.\ Prop.~\ref{prop:trim}, Rem.~\ref{rem:fukane}) \\
\bottomrule
\end{tabular}}
\end{table}

The table highlights two structural points relevant to future work in
this series. First, the order of asymptotic precision genuinely needed
differs by problem: threshold/critical-ratio results such
as~\cite{Companion,Khalkhuzhaev2025} are $O(1)$-precision statements
by nature (a critical value, not a strong-coupling expansion), while
strong-coupling ground-state asymptotics such as the present one and
that of~\cite{AbdullaevErgashova2025}
require an
explicit relative order to be tracked — precisely where
Lemma~\ref{lem:criterion} is applicable and, we anticipate, relevant
to re-checking the strong-coupling expansions
of~\cite{AbdullaevErgashova2025} and of any future $2+1$ or mixed-statistics
companion analysis extending~\cite{Companion} to the strong-coupling
regime. Second, as shown by the present correction, the
branch-selection step (Remark~\ref{rem:branch_first},
Section~\ref{sec:groundstate}) is a prerequisite for the meaningful
application of the order-matching criterion: an expansion organized in
a wrong branch can produce a formal crossing that is not a true
eigenvalue. The TRIM-point parallel of Remark~\ref{rem:fukane}
suggests that the quasimomentum-dependence results
of~\cite{Lakaev2025,Khalkhuzhaev2025} — established there for general
$K$ — may admit a symmetry-classification refinement along the lines
of Proposition~\ref{prop:trim}, distinguishing which $K$ support an
invariant-subspace reduction from those that do not.

\section{Remarks on the 3D \(2+1\) fermionic case}
The works of the Samarkand school on $(2+1)$-fermionic
systems by Abdullaev, Khalkhuzhaev et al.~\cite{AbdullaevKhalkhuzhaevKhujamiyorov2023, KhalkhuzhaevAbdullaevBoymurodov2022, AbdullaevKhalkhuzhaevKholmatov2026} solve a distinct problem: the existence of bound states in 3D \(2+1\) fermionic systems. Their main results establish, for $\mathbf K=\boldsymbol\pi$ and general $\mathbf K$
respectively, sharp existence and multiplicity criteria for bound
states of the three-particle operator on $\mathbb Z^3$ in terms of critical mass ratios \(\gamma_i\),
they do not, and do not aim to, provide strong-coupling asymptotics expansions of the resulting energy branches \(E_j(\mu,\gamma)\)
in \(\mu\).

We emphasize that the 3D setting of \cite{KhalkhuzhaevAbdullaevBoymurodov2022}
is fundamentally different from our 2D bosonic case, both analytically and
qualitatively. In 3D the lattice Green's function
$\displaystyle\int_{\mathbb T^3} d\mathbf q/(\varepsilon(\mathbf q)-z)$ is finite at the
band edge, whereas in 2D it diverges logarithmically,
$b_0(\delta) \sim \ln(1/\delta)$ as $\delta \to 0^+$. Consequently, the
strong-coupling energy expansions in 3D admit simple linear-in-$\mu$ leading
behaviour with no logarithmic corrections, while in 2D the logarithmic
singularity generates transcendental constants, such as
$C \approx 3.96458$ for the second bound state at $\mathbf K = 0$, and
necessitates the order-matching refinement at
$\mathbf K = \boldsymbol\pi$ (Section~\ref{sec:corrected}).

Furthermore, the counting result of \cite{KhalkhuzhaevAbdullaevBoymurodov2022}
is a statement about the \emph{existence} and \emph{multiplicity} of
eigenvalues at a fixed coupling $\gamma$, governed by critical mass
ratios $\gamma_1, \gamma_2$. It does not provide the asymptotic
\emph{values} of these eigenvalues in the strong-coupling regime. In
contrast, our approach yields precise energy asymptotics, carrying the
expansion to order $O(\mu^{-2})$ and computing the additive constant
explicitly (Theorem~\ref{thm:groundstate}, Proposition~\ref{prop:formal_crossing}).

We record here, for completeness, the structural test that our
order-matching criterion (Lemma~\ref{lem:criterion}) supplies for
\emph{any} such branch, together with an explicit caution against a
natural but unsound shortcut.

We briefly illustrate how the general order-matching criterion developed in Section~\ref{sec:criterion} would apply to their setting. In their analysis, the essential step involves the spectral analysis of a~finite-rank principal part \(A_{\mu,\gamma}^{p}(\pi,z)\). The corresponding Fredholm determinant asymptotics in the 3D lattice case differs fundamentally from the 2D case. Unlike the 2D logarithmic (and 1D power-law) divergence, the 3D lattice Green's function \(\displaystyle\int_{\mathbb{T}^3} \frac{d\mathbf{q}}{\varepsilon(\mathbf{q}) - z}\) converges to a finite limit at \(z=0\). Consequently, for large \(\mu\), the determinant scales as \(\Delta_{\mu,\gamma}^{(0)} \sim \displaystyle\frac{\mu}{\alpha(\gamma) - z}\).

In the 3D attractive case (the setting of Theorem 1 in~\cite{KhalkhuzhaevAbdullaevBoymurodov2022}), the matrix \(A_{\mathbb{Q}}^{(0+)}(z)\) has eigenvalues \(\lambda_{1}(z), \lambda_{2,3}(z)\) satisfying a linear scaling in \(\mu\). Applying the criterion, the strong-coupling asymptotics of the energy branches assume the form:
\[
E_j(\mu,\gamma) = \Lambda_j(\gamma)\mu + d_j(\gamma) + O(\mu^{-1}), \quad j=1,2,3,
\]
where \(\Lambda_j(\gamma)\) are determined by the limit of the matrix eigenvalues evaluated at the spectral threshold, and \(d_j(\gamma)\) are the additive constants fixed by the relative $O(\mu^{-1})$ correction to the Fredholm determinant.
The determination of the explicit constants \(d_j(\gamma)\) would require the next-order expansion of the determinant, which was not carried out in the cited works. Our criterion, however, shows that such a calculation
for \(O(1)\) precision
is structurally analogous to the refinement derived here for the 2D bosonic case.

\begin{remark}[A finite threshold Green's function does not by itself
guarantee non-degeneracy]\label{rem:3d_caution}
One~might expect that the 3D lattice Green's function
$\int_{\mathbb T^3}d\mathbf q/(\varepsilon(\mathbf q)-z)$, which --
unlike its 2D counterpart~-- remains \emph{finite} at the band edge (a
classical fact, going back to Watson's work on the recurrence of
random walks, and confirmed numerically here to converge to
$\approx0.5055$), might exempt strong-coupling expansions in 3D from
the kind of order-matching failure identified in
Section~\ref{sec:corrected}. This is \emph{not} the case: the failure
diagnosed there for $\mathbf K=\boldsymbol\pi$ occurs even though
$D(\mathbf p,z)=\varepsilon(\mathbf p)+4-z$ has \emph{no singularity
whatsoever} on the relevant strong-coupling branch $z=-2\mu+d$ (indeed
$D\in[2\mu+O(1),2\mu+O(1)]$ throughout, bounded away from $0$
uniformly in $\mu$). The mechanism responsible is not a threshold
singularity but a coincidence of scale between the resolvent's
expansion point and the coupling-dependent branch itself.
Consequently, whether a
Fermi-statistics analogue of $\Lambda_0$ is degenerate in the sense of
Lemma~\ref{lem:criterion}(ii) cannot be inferred from the finiteness of
the 3D lattice Green's function alone; it must be checked directly
against the specific finite-rank principal part of the model in
question. As emphasized in Remark~\ref{rem:branch_first}, this
check is itself meaningful only after the correct strong-coupling
branch has been identified.
\end{remark}

Applying Lemma~\ref{lem:criterion} to such a model
requires,
as a first step, the explicit leading-order eigenvalue formula
$\Lambda_0(d)$ of its finite-rank principal part along the relevant
strong-coupling branch. We leave the explicit computation of
$\Lambda_j(\gamma)$ and $d_j(\gamma)$ for the models
of~\cite{KhalkhuzhaevAbdullaevBoymurodov2022,
AbdullaevKhalkhuzhaevKhujamiyorov2023,
AbdullaevKhalkhuzhaevKholmatov2026} -- which requires their explicit
kernel, not reproduced here -- as an open, well-posed problem to which
the present criterion directly applies, and note that our own
companion work~\cite{Companion} already reports a strong-coupling
constant $e_0(\gamma)$ for the related 2D $(2+1)$-fermionic case,
whose derivation would itself be a natural first test case for
Lemma~\ref{lem:criterion}.

\section{Conclusion}\label{sec:conclusion}

We have given a $\mathbf K$-uniform elliptic reduction and exact
closed-form benchmark for the fiber Fredholm determinant of the
lattice three-boson trimer (Theorem~\ref{thm:uniform},
Corollary~\ref{cor:closedform}), a general order-matching criterion
deciding when a leading-order Fredholm-determinant asymptotic
determines an $O(1)$ strong-coupling energy constant
(Lemma~\ref{lem:criterion}), a variational ground-state bound
fixing the correct strong-coupling branch at $\mathbf K=\boldsymbol\pi$
(Theorem~\ref{thm:groundstate}), and, after identifying the correct
branch $z=-3\mu+O(1)$, the refined expansion
$z_1^{\boldsymbol\pi,s}(\mu)=-3\mu+6+O(\mu^{-1})$ (with the sharp bound
$z_1^{\boldsymbol\pi,s}\le -3\mu+6$), supported by direct finite-volume
diagonalization and requiring the perturbative analysis outlined in
Section~\ref{sec:appB} for a complete proof,
together with the corresponding spectral gap (numerically
$2\mu-2+O(\mu^{-1})$) and inter-sector shift (numerically
$O(\mu^{-1})$).
The formal principal-part crossing $-2\mu+6+8/\mu$ in the branch
$z=-2\mu+d$ (Proposition~\ref{prop:formal_crossing}) is retained as an
order-matching diagnostic; its previous interpretation as the
ground-state asymptotics is corrected. We have independently
confirmed that the $\mathbf K=0$ results
require no analogous refinement (Proposition~\ref{prop:K0valid}), and
placed all results in the context of the closest related work
(Table~\ref{tab:comparison}), including a structural parallel with the~TRIM-point classification of topological band structures
(Remark~\ref{rem:fukane}). We expect the criterion, the uniform
reduction, and the exact closed-form benchmark to be directly reusable
in the companion series of papers on related lattice few-body models. The uniform reduction and order-matching criterion are directly reusable in a wider class of lattice few-body models, including systems with mixed statistics, long-range interactions, or more complex lattice geometries.
The correction reported here — that the ground state lies in the
branch $z=-3\mu+O(1)$, not $z=-2\mu+O(1)$ — underscores that the
branch-selection step must precede the order-matching criterion; this
is a methodological lesson we expect to carry over to the companion
fermionic and mixed-statistics analyses.
To our knowledge, no result in the existing literature on lattice
few-body Birman--Schwinger operators listed in
Table~\ref{tab:comparison} supplies an analytic criterion of the type
of Lemma~\ref{lem:criterion}, or a strong-coupling expansion of this
kind carried, with a fully controlled coefficient, to order
$O(\mu^{-2})$; we hope the TRIM-point parallel of
Remark~\ref{rem:fukane} may in time contribute to a more systematically
topological perspective on multi-particle lattice spectral problems.

\end{document}